\documentclass{article}

\usepackage{arxiv}

\usepackage[utf8]{inputenc} 
\usepackage[T1]{fontenc}    
\usepackage{hyperref}       
\usepackage{url}            
\usepackage{booktabs}       
\usepackage{amsfonts}       
\usepackage{nicefrac}       
\usepackage{microtype}      
\usepackage{lipsum}

\usepackage{amsmath,amssymb,graphicx}

\newtheorem{axiom}{Axiom}[section]

\newtheorem{theorem}[axiom]{Theorem}

\newcommand{\qed}{~~$\Box$}
\newcommand{\pulley}{{\sc Subset-Sum-Interval}}
\newcommand{\IPG}{IPG}
\newcommand{\NE}{NE}
\newcommand{\boxxx}[1]
 {\begin{center}\fbox{\begin{minipage}{12.00cm}#1\smallskip\end{minipage}}\end{center}}

\title{A note on the complexity of integer programming games}

\author{
  Margarida Carvalho
  \\
  CIRRELT and D\'epartement d'informatique et de recherche op\'erationnelle\\
   Universit\'e de Montr\'eal\\
 Montreal, QC H3C 3J7 \\
  \texttt{carvalho@iro.umontreal.ca} \\
}

\begin{document}
\maketitle

\begin{abstract}
In this brief note, we prove that the existence of Nash equilibria on integer programming games is $\Sigma^p_2$-complete.
\end{abstract}

\keywords{Computational Complexity \and Game Theory \and Integer Programming \and Nash Equilibria}

\section{Introduction}

Integer programming games ({\IPG}s) model games in which there is a finite set of \emph{players} $M=\lbrace 1,2 \ldots, m \rbrace$ and   each player $p \in M$ has a set of \emph{feasible strategies} $X^p$ given by  lattice points inside a polytopes described  by  finite systems  of  linear  inequalities. Therefore, each players aims to solve 
\begin{equation}
  \max_{x^p \in X^P}  \quad  \Pi^p(x^p,x^{-p}),
\label{GeneralProblem}%
\end{equation}
where $x^p$ is player $p$'s strategy and $x^{-p}$ is the vector of strategies of all players, except player $p$.

A vector $x \in \prod_{p \in M} X^p$ is called \emph{a pure profile of strategies}.  If a pure profile of strategies $x$ solves the optimization problem~(\ref{GeneralProblem}) for all players, it is called \emph{pure Nash equilibrium}.  A game may fail to have pure equilibria and, therefore, a broader  solution concept for a game must be introduced, the \emph{Nash equilibria}. Under this concept, each player can assign probabilities to her pure strategies.  Let $\Delta^p$ denote the space of Borel probability measures over $X^p$ and $\Delta = \prod_{p \in M} \Delta^p$. Each player $p$'s expected payoff for a profile of strategies $\sigma \in \Delta$ is
 \begin{equation}
 \Pi^p(\sigma) = \int_{X^p} \Pi^p(x^p,x^{-p}) d \sigma.
\label{expected_utility}
\end{equation}
 A \emph{Nash equilibrium} (\NE) is a profile of strategies $\sigma  \in \Delta$ such that
\begin{equation}
 \Pi^p(\sigma)  \geq \Pi^p(x^p, \sigma^{-p}),  \qquad \forall p \in M \qquad \forall x^p \in X^p .
\label{NE_definition}
\end{equation}

In~\cite{Carvalho2018}, the authors discuss the existence of Nash equilibria for integer programming games. It is proven that the existence of pure Nash equilibria for {\IPG}s is  $\Sigma^p_2$-complete and that even the existence of Nash equilibria is  $\Sigma^p_2$-complete. However, the later proof seems incomplete in the ``proof of only if'', since it does not support  why we can conclude that the leader cannot guarantee a payoff of 1. In the following section, we provide a correct proof using a completely new reduction. 

\section{Computational Complexity}
In what follows, we show that even in the simplest case, linear integer programming games with two players, the existence of Nash equilibria is a  $\Sigma^p_2$-complete problem.

\begin{theorem}
The problem of deciding if an {\IPG} has a Nash equilibrium is $\Sigma^p_2$-complete problem.
\end{theorem}

\begin{proof}
The proof that this decision problem belongs to  $\Sigma^p_2$ can be found in~\cite{Carvalho2018}.

It remains to show that it is  $\Sigma^p_2$-hard. We will reduce the following to  $\Sigma^p_2$-complete problem to it (see\cite{EgWo2012}):

\boxxx{
\vspace{0.2cm}
\textbf{Problem: {\pulley}}

\vspace{0.2cm}
\textbf{INSTANCE } A sequence $q_1,q_2,\ldots,q_k$ of positive integers; two positive integers $R$ and $r$ with $r\leq k$.

\vspace{0.2cm}
\textbf{QUESTION }Does there exist an integer $S$ with $R\leq S<R+2^r$ such that none of the
subsets $I\subseteq\{1,\ldots,k\}$ satisfies $\sum_{i\in I}q_i=S$?
}

Our reduction starts from an instance of {\pulley}. We construct the following instance of {\IPG}
\begin{itemize}
\item The game has two players, $M=\lbrace Z, W \rbrace$.
\item Player $Z$ solves
\begin{subequations}
\begin{alignat}{4}
  \max_{z} &  \quad   \frac{1}{2}z_0+\sum_{i=1}^k q_i z_i + Qz(2w-z)\\
  s.t.  &\quad  \frac{1}{2}z_0+\sum_{i=1}^k q_i z_i \leq z\\
  & \quad z_0, z_1, \ldots, z_k \in \lbrace 0,1 \rbrace\\
  & \quad R \leq z \leq R+2^r-1, z \in \mathbf{N}
\end{alignat}
\label{PlayerZ}%
\end{subequations}

Add binary variables $y \in \lbrace 0,1\rbrace^r$ and make $z =R + \sum_{i=0}^{r-1} 2^i y_i$. Note that $z^2= Rz+\sum_{i=0}^{r-1} 2^i y_iz$. Thus, replace $y_iz$ by $h_i$ and add the respective McCormick constraints~\cite{McCormick76}. In this way, we can equivalently linearize the previous problem:
\begin{subequations}
\begin{alignat}{4}
  \max_{z,y,h} &  \quad   \frac{1}{2}z_0+\sum_{i=1}^k q_i z_i +2 Qzw-QRz-\sum_{i=0}^{r-1} 2^i  h_i\\
  s.t.  &\quad  \frac{1}{2}z_0+\sum_{i=1}^k q_i z_i \leq z\\
  & \quad z_0, z_1, \ldots, z_k \in \lbrace 0,1 \rbrace\\
  & \quad R \leq z \leq R+2^r-1, z \in \mathbf{N}\\
  & \quad z = R+ \sum_{i=0}^{r-1} 2^i y_i\\
  & \quad y_0, y_1, \ldots, y_{r-1} \in \lbrace 0,1 \rbrace\\
  & \quad h_i \geq 0 &i=0,\ldots,r-1\\
  & \quad h_i \geq z+(R+2^r-1)(y_i -1) & i=0,\ldots,r-1 \\
  & \quad h_i \leq z+R(y_i-1) & i=0,\ldots,r-1\\
  & \quad h_i \leq (R+2^r-1)y_i & i=0,\ldots,r-1
\end{alignat}
\label{PlayerZlin}%
\end{subequations}
For sake of simplicity, we consider the quadratic formulation (\ref{PlayerZ}). The linearization above serves the purpose of showing that the proof is valid even under linear utility functions for the players. 
\item Player W solves 
\begin{subequations}
\begin{alignat}{4}
  \max_{w} &  \quad   (1-z_0)w_0\\
  s.t.  & \quad R \leq w \leq R+2^r-1, z \in \mathbf{N}\\
  &w_0 \in \mathbf{R}
\end{alignat}
\label{PlayerW}%
\end{subequations}
\end{itemize}

(Proof of if).  Assume that the {\pulley} instance has answer YES. Then, there is  $R\leq S<R+2^r$  such that for all subsets $I\subseteq\{1,\ldots,k\}$, we have $\sum_{i\in I}q_i\neq S$. Let player $W$ strategy be $w^*=S$ and $w_0^*=0$. Note that the term $Qz(2w-z)$ in player $Z$'s utility is dominant and attains a maximum when $z$ is equal to $w$. Thus, make $z^*=w^*=S$ and since $\sum_{i=1}^k q_i z_i$ is at most $S-1$, make $z_0^*=1$. Choose $z_i^*$ such that the remaining utility of player $Z$ is maximized. By construction, player $Z$ is selecting her best response to $(w^*,w_0^*)$. Sinze $z_0^*=1$, then player $W$ is also selecting an optimal strategy. Therefore, we can conclude that there is an equilibrium.

(Proof of only if). Assume that the {\pulley} instances has answer NO. Then, for all $R\leq S<R+2^r$, there is a subset $I\subseteq\{1,\ldots,k\}$ such that $\sum_{i\in I}q_i=S$. In this case, player $Z$ will always make $z_0=0$ which gives incentive for player $W$ to choose $w_0$ as  large as possible. Since $w_0$ has no upper bound, there is no equilibrium for the game.
\qed
\end{proof}

\section*{Acknowledgements}

The authors wish to thank Sriram Sankaranarayanan for bringing to our attention the incompleteness of our previous proof in~\cite{Carvalho2018}. 

We wish acknowledge the support of the Institut de valorisation des donn\'ees and  the Canadian Natural Sciences and Engineering Research Council  under the discovery grants program.

\bibliographystyle{unsrt}  


\end{document}